\documentclass[10pt,onecolumn]{elsarticle}
\usepackage{geometry}                
\usepackage{graphicx}
\usepackage{amssymb}
\usepackage{amsmath}
\usepackage{epstopdf}
\usepackage{array}
\usepackage{amsthm}

\usepackage{array}
\usepackage{multirow,array}
\usepackage{color}
\usepackage{url}

\theoremstyle{plain}
\newtheorem{thm}{Theorem}[section] 

\theoremstyle{definition}
\author[add1]{Michael S. Harr\'e}
\ead{michael.harre@sydney.edu.au}
\author[add2]{Adam Harris}
\ead{aharris5@une.edu.au}
\author[add3]{Scott McCallum}
\ead{scott.mccallum@mq.edu.au}

\address[add1]{Complex Systems Research Group, Faculty of Engineering and IT, \\ The University of Sydney, Sydney, NSW 2006}
\address[add2]{School of Science and Technology, \\
University of New England, Armidale, NSW 2351}
\address[add3]{Department of Computing \\
Macquarie University, NSW 2109}

\title{Singularities and Catastrophes in Economics: Historical Perspectives and Future Directions}
\begin{document}

\begin{abstract}
Economic theory is a mathematically rich field in which there are opportunities for the formal analysis of singularities and catastrophes. This article looks at the historical context of singularities through the work of two eminent Frenchmen around the late 1960s and 1970s. Ren\'e Thom (1923-2002) was an acclaimed mathematician having received the Fields Medal in 1958, whereas G\'erard Debreu (1921-2004) would receive the Nobel Prize in economics in 1983. Both were highly influential within their fields and given the fundamental nature of their work, the potential for cross-fertilisation would seem to be quite promising. This was not to be the case: Debreu knew of Thom's work and cited it in the analysis of his own work, but despite this and other applied mathematicians taking catastrophe theory to economics, the theory never achieved a lasting following and relatively few results were published. This article reviews Debreu's analysis of the so called {\it regular} and {\it crtitical} economies in order to draw some insights into the economic perspective of singularities before moving to how singularities arise naturally in the Nash equilibria of game theory. Finally a modern treatment of stochastic game theory is covered through recent work on the quantal response equilibrium. In this view the Nash equilibrium is to the quantal response equilibrium what deterministic catastrophe theory is to stochastic catastrophe theory, with some caveats regarding when this analogy breaks down discussed at the end.
\end{abstract}
\maketitle

\section{Introduction: A key point in the history of economics ``in the mathematical mode''}

To single out a specific point in the history of a field of study and to say: {\it Here is where it all began} oversimplifies the complicated relationships between competing paradigms. However we can look at specific lines of research that have, with hindsight, become the dominant paradigm and look to what those authors wrote at the time to justify their specific view point in order to understand how an influential researcher framed their point of view. One such researcher is G\'erard Debreu, an economist who was the sole recipient of the 1983 Nobel Prize in Economics\footnote{Formally: {\it The Swedish National Bank's Prize in Economic Sciences in Memory of Alfred Nobel.}} for ``having incorporated new analytical methods into economic theory and for his rigorous reformulation of general equilibrium theory''. His work was an important step in the mathematisation of economics~\cite{debreu1987theory} and his Nobel speech was titled: {\it Economic Theory in the Mathematical Mode}~\cite{debreu1984economic}. Debreu acknowledged~\cite{debreu1976regular,debreu1984economic} the importance of Thom's work on catastrophe theory~\cite{thom1976structural,thom1977structural}, but chose not to incorporate these ideas into his axiomatic treatment of economics. At the same time applied mathematicians were using catastrophe theory in many fields in a qualitative fashion that would ultimately lead to a backlash against catastrophe theory. The debate within the field of economic applications of catastrophe theory has been thoroughly reviewed in~\cite{rosser2007rise} \\

A common approach to simplifying the complexities of an entire economy is to reduce to a simplified form and study the local behaviour of that economy~\cite{debreu1970economies}: {\it It suffices to construct an ... economy with two commodities and two consumers ... and having several equilibria. There is a neighbourhood of that economy in which every economy has the same number of equilibria.} Simplifications like this allow us to use game theory with two economic agents each having two choices each, see~\cite{dasgupta2015debreu} for a recent treatment and discussion. \\

\begin{figure}[!ht]
\center
\includegraphics[width=.9\columnwidth]{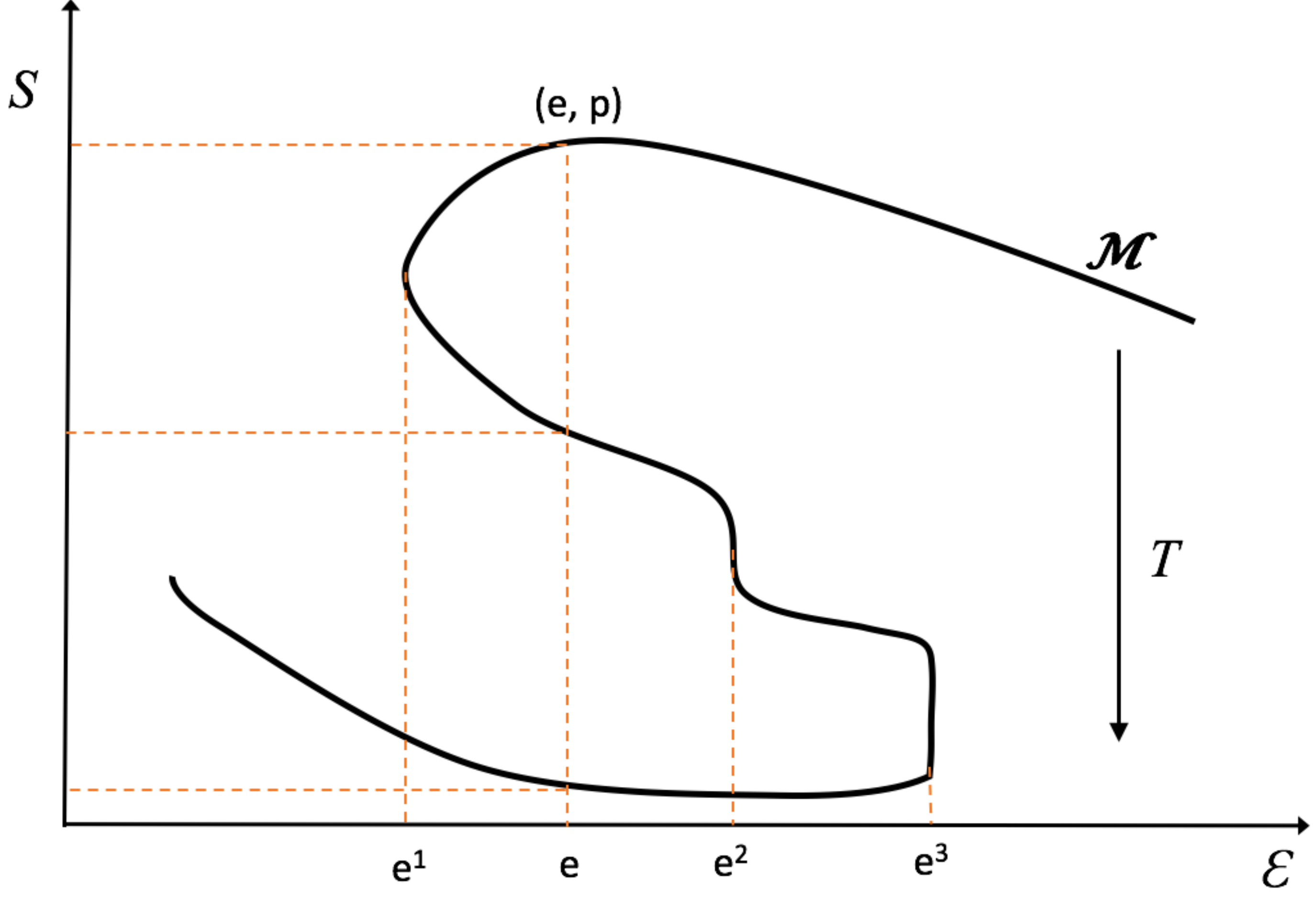}
\caption{\label{debreu_M} Reproduced after Figure 5. in~\cite{debreu1976regular}. The $(e,p)$ represents the economic configuration $e$ and the prices $p$ for the goods produced within $e$ that result in an equilibrium $s^* \in S$. As the figure suggest $s^*$ may not be unique but it is `almost always' locally unique, see the main text.} 
\end{figure}

An important contribution of Debreu was the distinction between regular and critical economies. In Figure~\ref{debreu_M} the different configurations of an economy are labelled $e$, $e^1$, $e^2$, and $e^3$ with the corresponding equilibrium states lying on the equilibrium surface $\mathcal{M}$. $T$ is the projection of the equilibrium surface $\mathcal{M}$ onto $\mathcal{E}$, the space of all economic configurations. Debreu showed that, apart from a subset of $\mathcal{E}$ with Lebesgue measure zero, there is an inverse mapping $T^{-1}: \mathcal{E} \rightarrow \mathcal{S}$ and in particular if $T^{-1}(e) \in \mathcal{M}$ is a locally unique equilibrium point then $e$ is a regular economy with at least one equilibrium point $s^*$. This follows from the inverse function theorem when the determinant of the Jacobian is away from 0. The special cases where the Jacobian is degenerate, i.e. $e^1$, $e^2$, and $e^3$, are the critical economies. To quote Debreu~\cite[pg. 284]{debreu1976regular}: {\it For instance the economy $e^1$ has a discrete set of (two) equilibria, but a continuous displacement of the economy in a neighbourhood of $e^1$ produces a sudden change in the set of equilibria.} The parallels with singularity theory are apparent and were known to Debreu. \\

The purpose of this article is to introduce some of the formal aspects of game theory, equilibrium theory, and their stochastic variations and to place these in the context of catastrophe theory. The article is laid out as follows. Section 2 establishes the key equilibrium result of Nash and reviews an important method for the explicit computation of Nash equilibria. Section 3 introduces singularities from the perspective of catastrophe theory, its extension to stochastic catastrophe theory, and their relation to bifurcations in game theory. Section 4 is a discussion of these results in the context of other related approaches and results.

\section{Nash Equilibrium \label{NE_sec}}

Optimisation leads naturally to singularities~\cite{arnol2003catastrophe} and the Nash equilibrium is a solution to an optimisation problem. As such the appearance of singularities is expected to be a generic property of game theory solutions. This aspect of optimisation has been explored by Arnol'd~\cite{arnold1984singularities}: {\it ... suppose we have to find $x$ such that the value of a function $f(x)$ is maximal ... On smooth change of the function the optimum solution changes with a jump, transferring from one competing maxima (A) to the other (B).} From this, singularities are expected to be a fundamental property of economic theory which is at its core a study of optimising the allocation of scarce resources. In this section we introduce the basic elements of game theory as a decision problem and the Nash equilibrium as an optimal solution.

\subsection{Basic definition and existence theorem}

Economics is concerned with situations in which there are $n$ agents, 
where each agent $i$ has a discrete set of finite choices $C_i =\{c_i^1, \ldots, c_i^{k_i}\}$. 
An economic game is a function from the choices agents make to a real valued payoff, 
one for each agent:
\begin{eqnarray}
g:C_1 \times \ldots \times C_n \rightarrow \mathbb{R}^n,
\end{eqnarray}
where $g_i$, the $i^{\mathrm{th}}$ element of $g$, 
is the payoff to agent $i$. 
A strategy for agent $i$ is a probability distribution over 
$C_i$: $$p_i: p_i(c_i^1), \ldots, p_i(c_i^{k_i})$$ such that 
$$\sum_{j_i=1}^{k_i} p_i(c_i^{j_i})=1.$$ 
To simplify notation we refer to $p_i(c_i^{j_i})$ as $p_i^{j_i}$. 
The definition of agent $i$'s expected payoff in the game $g$, which takes discrete choices as its arguments,
can be extended to take the strategies of the agents as arguments:
\begin{eqnarray}
g_i(p_i,p_{-i}) &=& \sum_{c_1^{j_1} \in C_1}\ldots \sum_{c_n^{j_n} \in C_n} g_i(c_1^{j_1}, \ldots, c_n^{j_n}) p_1^{j_1} \ldots p_n^{j_n} \label{exp_util}
\end{eqnarray}
where $-i$ denotes all indexed elements excluding element $i$. 
This definition of $g$ is an extension 
in the sense that it contains each discrete strategy 
as a special case in which one probability $p^{j_i}_i = 1$ and $p^{-j_i}_i = 0$ otherwise. 
These are called pure strategies. 
A strategy that is not a pure strategy is called a mixed strategy. 
Nash's equilibrium theory establishes the following result,
\begin{thm}{Nash Equilibrium:} There exists at least one $n$-tuple $(p^*_1, \ldots, p^*_n)$ such that:
\begin{eqnarray}
g_i(p^*_i,p^*_{-i}) & \geq & g_i(p_i,p^*_{-i}) \,\, \forall \, i, \, p_i. \label{NE_eqn}
\end{eqnarray}
\label{NE_Thm} \end{thm}
\begin{proof} (verbatim from Nash's original article~\cite{nash1950equilibrium}) Any $n$-tuple of strategies, one for each agent, may be regarded as a
point in the product space obtained by multiplying the $n$ strategy spaces of the agents. One such $n$-tuple counters another if the strategy of each agent in the countering $n$-tuple yields the highest obtainable expectation for its agent against the $n$-1 strategies of the other agents in the countered $n$-tuple. A self-countering $n$-tuple is called an equilibrium point. \\

The correspondence of each $n$-tuple with its set of countering $n$-tuples gives a one-to-many mapping of the product space into itself. From the definition of countering we see that the set of countering points of a point is convex. By using the continuity of the pay-off functions we see that the graph of the mapping is closed. The closedness is equivalent to saying: if $P_i, P_2, \ldots$ and $Q_1, Q_2, \ldots, Q_n, \ldots$ are sequences of points in the product space where $Q_n \rightarrow Q$, $P_n \rightarrow P$ and $Q_n$ counters $P_n$ then $Q$ counters $P$. \\

Since the graph is closed and since the image of each point under the mapping is convex, we infer from Kakutani's theorem that the mapping has a fixed point (i.e. point contained in its image). Hence there is an equilibrium point.
\end{proof}

Using the notation above this can be rephrased more formally~\cite{dasgupta2015debreu}: Given an agent $i$, for every ($n$-1)-tuple $p_{-i}$ of the strategies of the other agents, there is a best response function: 
\begin{eqnarray} 
\psi(p_{-i}) & = & \{p_i \, | \, g_i(p_i,p_{-i})\geq g_i(p'_i,p_{-i}) \,\, \forall \,p'_i\}.
\end{eqnarray}
This function consists of the set of strategies $p_i$ that maximise agent $i$'s payoff given the other agents' strategies $p_{-i}$. Define $\psi$ as the product space:
\begin{eqnarray}
\psi(p_1, \ldots, p_n) & = & \psi(p_{-1})\times \ldots \times \psi(p_{-n})
\end{eqnarray}
Then $\psi$ is an upper-hemicontinuous, convex valued, and non-empty-valued correspondence that maps the set of $n$-tuples to itself. Kakutani's fixed point theorem provides proof of the existence of a fixed point such that: $(p^*_1, \ldots, p^*_n) \in \psi(p^*_1, \ldots, p^*_n)$ that also satisfies Equation~\ref{NE_eqn}.

\subsection{Computational methods \label{comp_methods}}

This subsection briefly reviews some of the methods for explicitly computing Nash equilibria.
The main focus will be on two-agent games, which is the special case $n = 2$ of
the setting introduced in the previous subsection.
We shall suppose that agent 1 has $l$ pure strategies available, and agent 2 has $m$
pure strategies. We put $a_{i,j} = g_1(c_1^i, c_2^j)$ and
$b_{i,j} = g_2(c_1^i, c_2^j)$, 
and let $A$ and $B$ denote the $l \times m$ matrices $(a_{i,j})$ and $(b_{i,j})$, respectively.
To simplify the notation for the probability distributions we put
$x_i = p_1(c_1^i)$, for $1 \le i \le l$, and
$y_j = p_2(c_2^j)$, for $1 \le j \le m$, and we let $x$ and $y$ denote the column vectors
$(x_i)$ and $(y_j)$, respectively. We have
\begin{eqnarray}
\sum_{i=1}^l x_i = 1,~ x \ge {\bf 0}_l,~\sum_{j=1}^m y_j = 1,~ y \ge {\bf 0}_m.
\end{eqnarray}
Then the expected payoffs to agents 1 and 2 are expressed succinctly
as $x^TAy$ and $x^TBy$, respectively, where $x^T$ denotes the transpose of $x$.
A Nash equilibrium for such a game is then a pair $(x^*, y^*)$ satisfying (6)
and, for all $(x,y)$ satisfying (6), $(x^*)^T A y^* \ge x^T A y^*$ and
$(x^*)^T B y^* \ge (x^*)^T B y$.
That is, a Nash equilibrium point is a pair of mixed strategies that are best
responses to each other.\\

An equivalent characterization of a Nash equilibrium point $(x^*, y^*)$ is
that $(x^*, y^*)$ must satisfy (6) and
\begin{eqnarray}
\forall i~[x_i^* > 0 \Rightarrow (Ay^*)_i = \max_k (Ay^*)_k],
\end{eqnarray}
\begin{eqnarray}
\forall j~[y_j^* > 0 \Rightarrow ((x^*)^TB)_j = \max_k ((x^*)^TB)_k].
\end{eqnarray}
This characterization is proved in \cite{nash1951non} (see Equation 4).
It could be roughly expressed as follows: a mixed strategy is a best response to an
opponent's strategy if and only if it uses pure strategies that are best
responses amongst all pure strategies. See also \cite{codenotti2011computational} (Lemma 4.17).\\

The following result summarizes a further useful transformation of the problem
of finding a Nash equilibrium point.
\begin{thm}
Assume that all elements of $A$ and $B$ are positive.
Then there is a bijection $\phi$ between the set of all Nash equilibrium points $(x^*, y^*)$
and the set of all pairs $(u^*, v^*)$ of non-zero vectors 
$u^* \in \mathbb{R}^l$ and $v^* \in \mathbb{R}^m$, such that
\begin{eqnarray}
u^TB \le {\bf 1}_m, ~u \ge {\bf 0}_l, ~Av \le {\bf 1}_l, ~v \ge {\bf 0}_m,~{\rm and}
\end{eqnarray}                                      
\begin{eqnarray}
\forall i~[u_i^* > 0 \Rightarrow (Av^*)_i = 1],
\end{eqnarray}
\begin{eqnarray}
\forall j~[v_j^* > 0 \Rightarrow ((u^*)^TB)_j = 1].
\end{eqnarray}
\end{thm}
\begin{proof}
Let $(x^*, y^*)$ be a given Nash equilibrium point for the game $A$, $B$.
Then $(x^*)^T A y^*$ and $(x^*)^T B y^*$ are positive, by assumption that all elements of
$A$ and $B$ are positive.
So we may put $u_i^* = x_i^*/((x^*)^T B y^*)$ for each $i$, 
$v_j^* = y_j^*/((x^*)^T A y^*)$ for each $j$,
and define $\phi(x^*, y^*) = (u^*, v^*)$.
It is not difficult to verify that $(u^*,v^*)$ satisfies relations (9-11),
and each component of this ordered pair is non-zero.
Conversely, suppose that non-zero vectors $u^*$ and $v^*$ are given which satisfy
relations (9-11). We put $x^* = u^*/(u_1^* + \cdots + u_l^*)$,
$y^* = v^*/(v_1^* + \cdots + v_m^*)$, and define
$\psi(u^*, v^*) = (x^*, y^*)$.
Then $(x^*, y^*)$ satisfies relations (6-8).
It is straightforward to check that the mappings $\phi$ and $\psi$ are inverses of each
other. Hence $\phi$ is a bijection as claimed (as is its inverse $\psi$).
\end{proof}
We remark that the assumption of the above theorem that all elements of $A$ and $B$ are
positive is not too restrictive. For, if the given payoff matrices
$A$ and $B$ do not satisfy this assumption, then we could add to all entries of $A$ and $B$
a sufficienty large positive constant to ensure $A$ and $B$ become positive,
without changing the essential nature of the game.\\

We now focus on describing a method due to Lemke and Howson \cite{lemke1964equilibrium}
for finding a pair of non-zero vectors $(u^*, v^*)$ which satisfy relations (9-11).
We shall also call such a pair of vectors a Nash equilibrium point.
Hereafter we shall use $x$ and $y$ to denote $u$ and $v$, respectively,
and similarly for the corresponding starred variables.
This slight abuse of notation is to ensure consistency with most of the literature
on this subject. Consider the polyhedral sets
$X = \{x \in \mathbb{R}^l~|~x \ge {\bf 0},~x^T B \le {\bf 1}\}$ and
$Y = \{y \in \mathbb{R}^m~|~y \ge {\bf 0},~Ay \le {\bf 1}\}$
in $\mathbb{R}^l$ and $\mathbb{R}^m$, respectively.
The method of Lemke and Howson is conveniently described in terms of
a labelling process for the vertices of $X$ and $Y$, originally proposed by Shapley
\cite{shapley1974note}.
Let $I = \{1,2, \ldots, l\}$ and $J = \{l+1, l+2, \ldots, l+m\}$ be label sets
for agent 1 and agent 2, respectively, and put $K = I \cup J$.
With each vertex $x$ of $X$ we associate a set $L(x)$ of labels from $K$
as follows. For $1 \le i \le l$, $x$ is given the label $i$ if $x_i = 0$.
For $1 \le j \le m$, $x$ is given the label $l+j$ if $(x^T B)_j = 1$.
With each vertex $y$ of $Y$ we associate a set $M(y)$ of labels from $K$ as follows.
For $1 \le j \le m$, $y$ is given the label $l+j$ if $y_j = 0$.
For $1 \le i \le l$, $y$ is given the label $i$ if $(A y)_i = 1$. \\

We say that $X$ is {\em nondegenerate} if each vertex $x$ of $X$ satisfies exactly $l$
equations from amongst the
$x_i = 0$, with $1 \le i \le l$, and the
$(x^T B)_j = 1$, with $1 \le j \le m$. 
Similarly, $Y$ is {\em nondegenerate} if each vertex
$y$ of $Y$ satisfies exactly $m$ equations from amongst the
$y_j = 0$, with $1 \le j \le m$, and the
$(A y)_i = 1$, with $1 \le i \le l$.
So, if $X$ and $Y$ are both nondegenerate then each vertex of $X$ has exactly $l$ labels
and each vertex of $Y$ has exactly $m$ labels.
A vertex pair $(x,y)$ has at most $l+m$ different labels taken altogether, and this number
could be less than $l+m$ if $L(x)$ and $M(y)$ have nonempty intersection.
A two-agent game defined by $A$ and $B$ is also called {\em nondegenerate} if
the corresponding polyhedral sets $X$ and $Y$ are nondegenerate.
The description of the Lemke-Howson method is simplest in case the game under consideration
is nondegenerate, and we shall assume this hereafter.\\

Consider pairs of vertices $(x,y) \in X \times Y \subset \mathbb{R}^{l+m}$.
We say that $(x,y)$ is {\em completely labelled} if
$L(x) \cup M(y) = K$. The importance of this concept is due to the fact that
$(x,y) \in X \times Y$ is completely labelled if and only if either 
$(x,y) = ({\bf 0}, {\bf 0})$ or $(x,y)$ is a Nash
equilibrium point for $A$ and $B$ (by Equations 10 and 11, and the definition of the
labelling scheme). To avoid the special case, the origin $({\bf 0}, {\bf 0})$ is
termed the {\em artificial equilibrium}.
Now let $k \in K$. We say that $(x,y)$ is {\em $k$-almost completely labelled}
if if $L(x) \cup M(y) = K - \{k\}$. A $k$-almost completely labelled vertex
has exactly one duplicate label, that is, exactly one label which belongs to $L(x) \cap M(y)$.\\

The Lemke-Howson algorithm is succinctly, abstractly, but not quite explicitly, described
in terms of following a certain path through a graph $G$ related to $X \times Y$.
Indeed, let $G$ be the graph comprising the vertices and edges of $X \times Y$.
Fix some $k \in K$.
An equilibrium point $(x^*, y^*)$ (which is completely labelled, as mentioned above) is adjacent
in $G$ to exactly one $k$-almost completely labelled vertex $(x', y')$, namely, that vertex
obtained by dropping the label $k$.
Similarly, a $k$-almost completely labelled vertex $(x,y)$ has exactly two neighbours
in $G$ which are either $k$-almost completely labelled or completely labelled.
These are obtained by dropping in turn one component of the unique duplicate label
that $x$ and $y$ have in common.
These two observations imply the existence of a unique $k$-almost completely labelled path in $G$ 
from any one given equilibrium point to some other one. (The end points of such a path
are completely labelled, but all other vertices on the path are $k$-almost completely labelled.)
The Lemke-Howson algorithm starts at the artificial equilibrium $({\bf 0}, {\bf 0})$.
Choosing $k \in K$ arbitrarily,
it then follows the unique $k$-almost competely labelled path in $G$
from $({\bf 0}, {\bf 0})$, step by step,
until it reaches a genuine equilibrium point $(x^*, y^*)$ whereupon it terminates.
The following result summarizes the conclusions we may draw from this description.
\begin{thm}
Let $A$ and $B$ represent a nondegenerate game and let $k$ be a label in $K$.
Then the set of $k$-almost completely labelled vertices together with the completely labelled ones,
and the set of edges joining pairs of such vertices, consist of disjoint paths and cycles.
The endpoints of the paths are the Nash equilibria of the game, including the
artificial one $({\bf 0}, {\bf 0})$. The number of Nash equilibria of the game is therefore odd.
\end{thm}
The above result provides an alternative and constructive proof of Theorem 2.1
in the special case of two-agent nondegenerate games.
The algorithm can be described in a computationally explicit way using certain
concepts and techniques of the well known simplex method of linear programming.

Lemke and Howson \cite{lemke1964equilibrium} proposed perturbation techniques for dealing with
degenerate games. Eaves \cite{eaves1971linear} provided an explicit computational procedure
based on these ideas. A clear exposition of this procedure is contained in
\cite{mckelvey1996computation}.

For $n$-agent games, with $n > 2$, the problem of finding a Nash equilibrium is no longer
of linear character. Thus the Lemke-Howson algorithm cannot be applied directly.
Now a Nash equilibrium for an $n$-agent game can be characterized as a fixed point of
a certain continuous function from a product $\Delta$ of unit simplices into itself.
A tractable approach to computing an equilibrium point for such a game can then be based on a
path finding technique related to an algorithm of Scarf \cite{scarf1967approximation}
for finding fixed points of a continuous function defined on a compact set.
An exposition of such an approach, known as a {\em simplicial subdivision method},
is found in \cite{mckelvey1996computation}.

The methods recalled in this subsection concern computation of a single,
so called {\em sample}, equilibrium point for a game.
Such methods are generally not adaptable to the construction of {\em all}
equilibria of a given game, however.
Methods for locating all equilibria exist -- see section 6 of
\cite{mckelvey1996computation}.

Homotopy methods for computing equilibria have also been developed
\cite{herings2010homotopy}.

\section{Singularities in Economics}

In the introduction to Section~\ref{NE_sec} it was pointed out that a core element of economics is the study of applied problems in optimisation. Here this notion is made explicit by introducing catastrophe theory and its stochastic extension and then game theory and its stochastic extension emphasising the role of singularities in each. 

\subsection{Catastrophe theory}

Catastrophe theory was introduced as an approach to economic equilibria through the work of Zeeman in the 1970s~\cite{zeeman1974unstable} in which the attempt was made to explain the dynamics of financial markets using qualitative arguments based on Thom's earlier work in catastrophe theory. We follow~\cite{diks2016can} in establishing our general framework. We posit a potential function $G(\textbf{x}|\textbf{u}): \mathbb{R}^{n}\times \mathbb{R}^{k} \rightarrow \mathbb{R}$ of a vector of $n$ state variables $\mathbf{x}$ and a vector of $k$ control parameters $\mathbf{u}$~\cite{stewart1977catastrophe}. One of the simplest examples is a system in which there is one (time dependent) state variable $x_t$ and two control parameters $u_1$ and $u_2$. The dynamics of such a system are given by the ordinary differential equation:
\begin{eqnarray}
dx_t & = & -\frac{\partial G(x_t|u_1,u_2)}{\partial x_t} dt \label{dynamics}
\end{eqnarray} 
for which the stationary solutions of the system are given by: $\frac{\partial G(x_t|u_1,u_2)}{\partial x_t} = 0$. The following example is the starting point of recent analyses of asset markets~\cite{barunik2009can,diks2016can}: 
\begin{eqnarray}
G(x_t| u_1, u_2) & = & \frac{1}{4}x_t^4 - \frac{1}{2}u_1x_t^2 - u_2x_t. \label{quart_util}
\end{eqnarray} 
The stationary states of the system are the specific $x$ that solve the equation: 
\begin{eqnarray} 
\left.-\frac{\partial G(x_t | u_1, u_2)}{\partial x_t}\right|_{x} & = & -x^3 + u_1 x + u_2 \;\; =\;\;  0. \label{zeroes}
\end{eqnarray} 
Figure~\ref{cusp_cat} shows the time independent stationary points $x_t \equiv x$ as an equilibrium surface that is functionally dependent on the control parameters $u_1$ and $u_2$. \\ 

One of the most important results of Thom's work was the classification of catastrophes into one of seven {\it elementary catastrophes}, catastrophes for families of functions having no more than two state variables and no more than four codimensions. Following Stewart~\cite{stewart1977catastrophe} we first give two key definitions needed to state Thom's theorem, then the theorem is given, and an illustrative example then follows. \\

{\it Right equivalence}: Two smooth ($C^{\infty}$) functions $f:\mathbb{R}^n \rightarrow \mathbb{R}$ and $g:\mathbb{R}^n \rightarrow \mathbb{R}$ are said to be right equivalent if there is a local diffeomorphism $\phi: \mathbb{R}^n \rightarrow \mathbb{R}^n$ with $\phi(\mathbf{0}) = \mathbf{0}$ ($\mathbf{0}$ the zero vector) such that $f(\mathbf{x}) = g(\phi(\mathbf{x}))$ for all $\mathbf{x}$ in some neighbourhood of the origin. \\

{\it Codimension}: If $f: \mathbb{R}^n \rightarrow \mathbb{R}$ (the {\it germ} in the following) is smooth then the codimension of $f$ is the smallest $k$ for which there exists a $k$-dimensional smooth unfolding 
\begin{eqnarray}
G: \mathbb{R}^n \times \mathbb{R}^k \rightarrow \mathbb{R}
\end{eqnarray}
with
\begin{eqnarray}
G(\mathbf{x}\,|\,\mathbf{0}) & = & f(\mathbf{x}), \,\, \mathbf{x} \in \mathbb{R}^n
\end{eqnarray}
which is stable. If no such unfolding exists, the codimension is defined to be $\infty$. In this way the codimension measures the ``degree of instability'' of the function $f$.

\begin{thm} \cite{thom1975structural} Let $G:\mathbb{R}^{n}\times \mathbb{R}^{k} \rightarrow \mathbb{R}$ be a smooth stable family of functions each of which has a critical point at the origin, and suppose that $k \leq 4$. Set $f(\mathbf{x}) = G(\mathbf{x}\, |\, \mathbf{0})$ for all $\mathbf{x}$ in $\mathbb{R}^n$ near the origin. Then $f$ is right equivalent, up to a sign, to one one of the germs $f^{*}$ in Table~\ref{cat_table}.
\begin{table}[ht]
\centering
\caption{Thom's elementary catastrophes \label{cat_table}}
\begin{tabular}[t]{llcl}
\hline
Catastrophe name & Germ ($f^{*}$) & Codimension & Unfolding ($G^{*}$)\\
\hline
Fold & $x^3$ & 1 & $x^3 + ax$ \\
Cusp & $x^4$ & 2 & $x^4 + ax^2 + bx$ \\
Swallow tail & $x^5$ & 3 & $x^5 + ax^3 + bx^2 + cx$ \\
Butterfly & $x^6$ & 4 & $x^6 + ax^4 + bx^3 + cx^2 + dx$ \\
Hyperbolic umbilic & $x^3 + xy^2$ or $x^3 + y^3$ & 3 & $x^3 + y^3 + axy + bx + cy$ \\
Elliptic umbilic & $x^3 - xy^2$ & 3 & $x^3 - 3xy^2 + a(x^2+y^2) + bx + cy$ \\
Parabolic umbilic & $x^2y + y^4$ & 4 & $x^2y + y^4 + ax^2 + by^2 + cx + dy$ \\
\hline
\end{tabular}
\end{table}
\end{thm}

The $a,b,c,d$ terms are control variables and the $x$ and $y$ are state variables. Further, $G$ is right equivalent, up to sign, to the expression $G^{*}$ on the same line as $f^{*}$ in Table~\ref{cat_table} (where right equivalence of function families is defined in a similar sense to that of functions). In this table the catastrophes describe the geometry of the projection onto control parameter space of the surface defined by the partial derivative: $\frac{ \partial G^*}{\partial x} = 0$.\\

The unfolding of $x^4$ results in the cusp catastrophe. A surprising result is that the Taylor series does not need to converge on $G$ or indeed any function and so the importance of this proof lies in the fact that a truncated Taylor series expansion provides a qualitatively (topologically) correct description of a system given by Equation~\ref{dynamics}. This is the case because this family, as do all of the elementary catastrophes, has the stability property that for sufficiently small perturbations of $G(x|u_1, u_2)$ the topology is left unaffected. \\

Thom's approach can be illustrated with a simple example from~\cite{stewart1977catastrophe}. Suppose we are given a smooth stable potential function $G(x\,|\,u_1, u_2)$ with one state variable $x$ and two control variables $u_1$, $u_2$. Regarding $G$ as a function family parametrised by $u_1$ and $u_2$, suppose that each member of the family has a critical point at the origin, and that $G(x\,|\,0, 0) = x^4 \, + \,$ higher order terms in $x$. Then, according to Thom's theorem, the function $G(x\,|\,0,0)$ is right equivalent to $x^4$, and the function family $G(x\,|\,u_1,u_2)$ is right equivalent to the unfolding $x^4 + u_1 x^2 + u_2 x$, for (possibly) adjusted state and control variables $x$, $u_1$ and $u_2$.\\


When catastrophe theory is applied to specific problems in the applied fields, the following interpretation of Thom's theorem is very useful (after~\cite[pg. 151]{stewart1977catastrophe}. The germ $f^*$, when observed in applications, might be topologically unstable in the sense that small perturbations to the system might result in qualitatively different behaviours. So in practice, if we observe $f^*$ we should also expect to see the rest of its unfolding $G^*$ as well, and this unfolding will be a topologically stable and qualitatively complete description of the system. We note that Berry~\cite{berry1982universal} has made the rather fascinating observation that a `battle of the catastrophes' emerges as a parameter $t$ varies through the catastrophe set on the singularity surface. That this generates power-law tails in a very specific fashion may have applications not yet explored.

\begin{figure}[!ht]
\center
\includegraphics[width=.9\columnwidth]{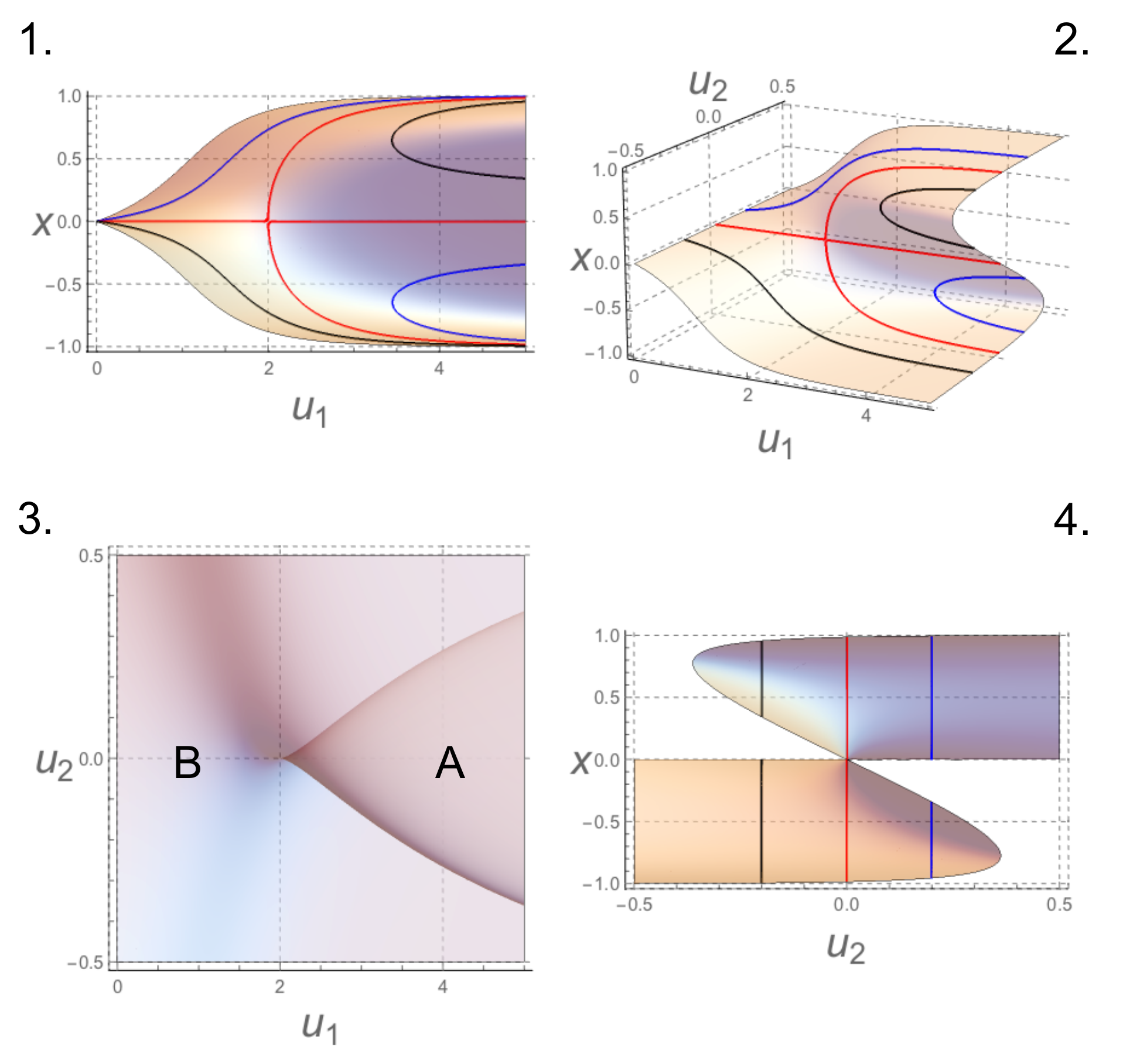}
\caption{\label{cusp_cat} The cusp catastrophe with three contours shown for fixed $u_2$ but variable $u_1$ control parameters. Plots 1, 2, and 4 show different perspectives of the equilibrium surface. Notice that the cross section $u_1=$ constant of Plot 4 has qualitatively similar properties to the manifold $\mathcal{M}$ of Figure~\ref{debreu_M}. Plot 2.: Two distinct regions $A$ and $B$ can be discerned in the projection of the equilibrium surface onto the control plane $\{u_1, u_2\}$, region $A$ has three equilibrium point and region $B$ has 1 equilibrium point. A Pitchfork bifurcation is shown in red and two fold bifurcations are shown in black and blue. Note that these plots correspond to the equilibrium surfaces (stationary states) that are the solutions to Equation~\ref{zeroes}.}
\end{figure}

\subsection{Stochastic Catastrophe Theory}

Catastrophe theory is a deterministic approach to modelling a system's dynamics but Cobb~\cite{cobb1978stochastic} was the first to extend catastrophe theory to stochastic differential equations which was later improved upon by Wagenmakers {\it et al}.~\cite{wagenmakers2005transformation}. The approach is to add noise to the evolutionary dynamics:
\begin{eqnarray}
dx_t & = & -\frac{\partial G(x_t|\mathbf{u})}{\partial x_t} dt  + \sigma(x_t) dW_t \label{SDE_cusp}
\end{eqnarray} 
where $-\frac{\partial G(x_t|\mathbf{u})}{\partial x_t} \equiv g(x_t)$ is the drift function, $W_t$ is a diffusion process and $\sigma(x_t)$ parameterises the strength of the diffusion process. The stationary probability distribution is given by~\cite{wagenmakers2005transformation}: 
\begin{eqnarray}
p(x|\mathbf{u}, \sigma(x)) & = & \mathcal{Z}^{-1}\exp\Big(2\int_{a}^{x_t} \frac{g(z) - 0.5(d_z\sigma(z)^2)}{\sigma(z)^2}dz\Big), \label{SDE_prob1} \\
p(x|\mathbf{u}, \xi) & = &  \mathcal{Z}^{-1}\exp\Big(\xi \int_{a}^{x_t} g(z)dz\Big)  \label{SDE_prob2}, \\
& = & \mathcal{Z}^{-1}\exp(G_{\sigma}(x)). \label{SDE_prob3}
\end{eqnarray} 
Here $\mathcal{Z}$ normalises the probability distribution, in Equation~\ref{SDE_prob1}: $d_z\sigma(z)^2 = \frac{d}{dz}\sigma(z)^2$, Equation~\ref{SDE_prob2} is a simplification in which: $2\sigma(z)^{-2} = \xi \equiv $ constant in $z$, and $G_{\sigma}(x) = \xi \int_{a}^{x_t} g(z)dz$ is the stochastic potential for the special case of $\sigma(z)$ being constant. \\

A potential function is necessary for catastrophe theory in order for there to be a gradient dynamic in Equation~\ref{dynamics}, this is equivalent to the symmetry of the Slutsky matrix~\cite{afriat2014demand} of economics. This symmetry is known to hold for {\it potential games}~\cite{sandholm2001potential} and so a potential function exists. The existence of a potential function is also equivalent to the existence of a Lyapunov function. If a (local) Lyapunov function does not exist then this is a significant obstacle to the use of catastrophe theory, however Thom has answered this objection~\cite{thom1977structural} by pointing out that near {\it any} attractor $A$ of {\it any} dynamical system there exists a local Lyapunov function. In evolutionary game theory, a dynamical extension to (static) economic game theory, a Lyapunov function can always be found for linear fitness functions that are of the type in Equation~\ref{exp_util}. The principal problem then is to confirm whether or not such a local Lyapunov function exhibits the bifurcation behaviour of catastrophe theory, and in general the answer is no~\cite{guckenheimer1973catastrophes}. Although this article covers what is properly called `elementary catastrophe theory', Thom hinted~\cite{thom1977structural} that there may be ways to accommodate these issues within the larger set of catastrophe theory that his methods encompass. It appears to still be an open question as to what extent catastrophe theory can be extended to accommodate these issues.

\subsection{Bifurcations in Nash equilibria}

To illustrate how bifurcations in the number of fixed points occur in the Nash equilibria of games we will use $G_2^2$ games, these games are described by two payoff matrices, one for each agent. Specifically, we consider {\it normal form}, two agent, {\it non-cooperative} games in which the agents $i = 1, 2$ select between one of two possible choices (pure strategies): $c_i^j \in C_i$. We refer to these games as $G_2^2$. The joint choices determine the utility for each agent $i$, $g_i: C \rightarrow \mathbb{R}$. The choices available to the agents and their subsequent payoffs are given by agent 1's payoff matrix:
{\renewcommand{\arraystretch}{2}
\begin{center}
    \begin{tabular}{cr|c|c|}
      & \multicolumn{1}{c}{} & \multicolumn{2}{c}{agent $2$}\\
      & \multicolumn{1}{c}{} & \multicolumn{1}{c}{$c^1_{2}$}  & \multicolumn{1}{c}{$c^2_{2}$} \\\cline{3-4}
      \multirow{2}*{agent $1$}  & $c^1_1$ & $g_1(c_1^1,c_2^1)$ & $g_1(c_1^1,c_2^2)$ \\\cline{3-4}
      & $c^2_1$ & $g_1(c^2_1,c_2^1)$ & $g_1(c^2_1,c_2^2)$ \\\cline{3-4}
    \end{tabular}
\end{center}}
\vspace{7mm}
and agent 2's payoff matrix:
{\renewcommand{\arraystretch}{2}
\begin{center}
    \begin{tabular}{cr|c|c|}
      & \multicolumn{1}{c}{} & \multicolumn{2}{c}{agent $2$}\\
      & \multicolumn{1}{c}{} & \multicolumn{1}{c}{$c^1_{2}$}  & \multicolumn{1}{c}{$c^2_{2}$} \\\cline{3-4}
      \multirow{2}*{agent $1$}  & $c^1_1$ & $g_2(c_1^1,c_2^1)$ & $g_2(c_1^1,c_2^2)$ \\\cline{3-4}
      & $c^2_1$ & $g_2(c^2_1,c_2^1)$ & $g_2(c^2_1,c_2^2)$ \\\cline{3-4}
    \end{tabular}
\end{center}}
\vspace{7mm}
Both matrices are often more compactly written as a single bi-matrix, for example the following bi-matrix is the well known {\it Prisoner's Dilemma} game where the vectors of joint payoffs are written $(g_1(c^i_1,c^j_2),g_2(c^i_1,c^j_2))$:
{\renewcommand{\arraystretch}{2}
\begin{center}
    \begin{tabular}{cr|c|c|}
      & \multicolumn{1}{c}{} & \multicolumn{2}{c}{agent $2$}\\
      & \multicolumn{1}{c}{} & \multicolumn{1}{c}{cooperate}  & \multicolumn{1}{c}{defect} \\\cline{3-4}
      \multirow{2}*{agent $1$}  & cooperate & $(-1,-1)$ & $(-3,0)$ \\\cline{3-4}
      & defect & $(0,-3)$ & $(-2,-2)$ \\\cline{3-4}
    \end{tabular}
\end{center}}
\vspace{7mm}
In the extended form of $g$, in which: 
\begin{eqnarray}
g: (g_1(p_1,p_2),g_2(p_1,p_2)) & \rightarrow & \mathbb{R}^2
\end{eqnarray} 
is an upper-hemicontinuous space\footnote{Upper-hemicontinuity does not hold if only discrete strategies can be used, and in order to use Kakutani's fixed point theorem it is a necessary condition that the agent's strategies are upper-hemicontinuous}, the bi-matrix representation is only useful in describing the discrete valued payoffs. \\

\begin{figure}[ht]
\center
\includegraphics[width=0.95\columnwidth]{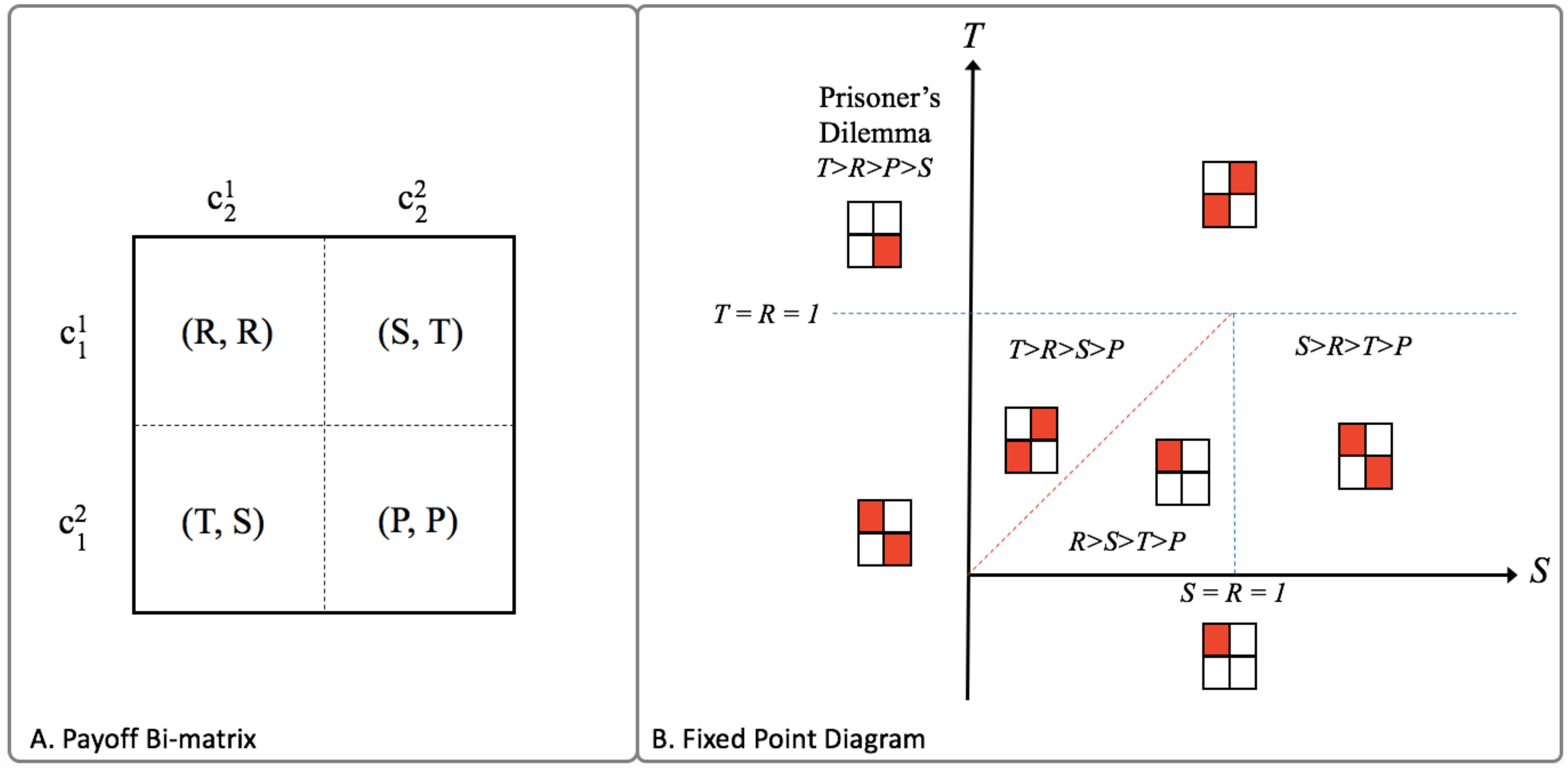}
\caption{\label{Bif_dia} {\bf A.} The payoff bi-matrix for a subset of $G_2^2$ games. The payoffs are symmetrical for simplicity, if the payoff matrix to agent 1 is $A_1$, then the payoff matrix for agent 2 is $A_2 = A_1^T$. Two parameters are held fixed: $R=1$ and $P=0$, leaving a two parameter ($T$ and $S$) family of games for which the number and location of the Nash equilibrium fixed points will vary as $T$ and $S$ vary. {\bf B.} A schematic diagram showing how the number and location of fixed points are a function of $T$ and $S$. The location of the pure strategy Nash equilibria are shown as red squares on a representation of the bi-matrix in each sector. For example the Nash equilibrium of the Prisoner's Dilemma is the pure strategy: $(p_1^2=1,p_2^2=1)$ which is located at the bottom right of the bi-matrix $(P,P)$ in the left diagram, and shown as a red square in the Prisoner's Dilemma region of the right diagram. Figure adapted from~\cite{hauert2001fundamental}}
\end{figure}

In order to see how singularities can occur in game theory we parameterise the payoff values in each agent's payoff matrix in such a way that allows for two control parameters and one state variable for each agent, i.e. the expected utility $g_i(p^*_i,p^*_{-i})$, at each Nash equilibrium, this approach is shown in Figure~\ref{Bif_dia}. Note that generally there are either one or three Nash equilibrium and Figure~\ref{Bif_dia} only shows the pure strategy Nash equilibria of the games. In the case where there are two pure strategy Nash equilibria then there is also a third mixed strategy Nash equilibria that is not shown in Figure~\ref{Bif_dia}.\\

To illustrate the change in the number of fixed points, begin by looking at the Prisoner's Dilemma game in the top left of the diagram. At this point $S<0$ but as $S$ increases and passes through $S=0$ two new pure strategy Nash equilibria form at the $(T,S)$ and $(S,T)$ pure strategies and a third mixed strategy equilibria forms that is not shown. Alternatively, beginning again from the Prisoner's Dilemma but decreasing $T$ from $T>1$ to $T<1$ the original pure strategy at $(P,P)$ remains but a new pure strategy equilibrium forms at $(R,R)$, as well as a mixed strategy equilibrium that is not shown. From this region of the game space, providing $1>T>0$, then letting $S$ increase from $S<0$ to $S>0$ we see that there remains three equilibrium points but the two pure strategy equilibria at $(R,R)$ and $(P,P)$ switch to $(S,T)$ and $(T,S)$. Other transitions in the state  space follow similar patterns. \\

\subsection{Quantal Response Equilibrium}

The quantal response equilibrium (QRE) is an extension of the Nash equilibrium concept developed by McKelvey and Palfrey~\cite{mckelvey1995quantal} in which agents do not perfectly optimise their choices, as in the Nash equilibrium, but instead there is some error in the choice they make, represented by stochastic uncertainty in their choices. The Nash equilibrium is recovered as a parameter representing the uncertainty in the decision tends to infinity, and so the Nash equilibria are a subset of the fixed points given by the QRE. \\

There are several different methods by which the QRE can be arrived at, McKelvey and Palfrey used differential topology and recent work by one of the authors used the method of maximising the entropy~\cite{wolpert2012hysteresis,harre2014strategic}. For our purposes we will simply define the relevant terms and state the Logit functional form of the QRE. We note that there is nothing special in this form of the QRE correspondence. It is to be expected that all of the work in the current article, in particular the analysis of bifurcations parameterised by a noise term $\beta$, will carry across to any regular QRE functional form. \\

Agent $i$'s expected utility $g_i(p_i,p_{-i})$ can be said to be {\it conditional} on $i$'s discrete choice $c_i^j$: $g_i(p_{-i}|c_i^j)$. The interpretation of $g_i(p_{-i}|c_i^j)$ is that it is the expected utility to agent $i$ if they choose $c_i^j$, i.e. they fix $p_i^j = 1$, while all of the other agents maintain a (possibly mixed) joint strategy $p_{-i}$. The definition of the equilibrium points given by the QRE are then the joint distributions $(p^*_1, \ldots, p^*_n)$ given by:
\begin{eqnarray}
p^*_i(c^j_i | \beta_i ) & = & \frac{\exp(\beta_i g_i(p^*_{-i}|c_i^j))}{\sum_{j} \exp(\beta_i g_i(p^*_{-i}|c_i^j)))} \label{QRE_1}
\end{eqnarray}
This satisfies the criteria of a probability distribution over agent $i$'s space of $j$ choices and the exponentiated function $\beta_i g_i(p^*_{-i}|c_i^j)$ is the product of a control parameter $\beta_i$ specific to each agent and a gradient $g_i(p^*_{-i}|c_i^j)$, cf. Equations~\ref{SDE_prob1}-\ref{SDE_prob3}. The parameter $\beta_i \in [0,\infty]$ controls the level of noise or uncertainty the agent has in selecting each strategy, when $\beta_i = 0$ the agent selects uniformly across their choices and when $\beta_i \rightarrow \infty$ the Nash equilibria of the game are recovered. For intermediate values of $\beta_i$ the agent prefers (up to some statistical uncertainty) one strategy over another only if it has a higher payoff, assuming all other agents are playing a known distribution over their own strategies. The probability is conditional on $\beta_i$ to remind us that the distribution has one free control parameter. \\

\begin{figure}[ht]
\center
\includegraphics[width=0.95\columnwidth]{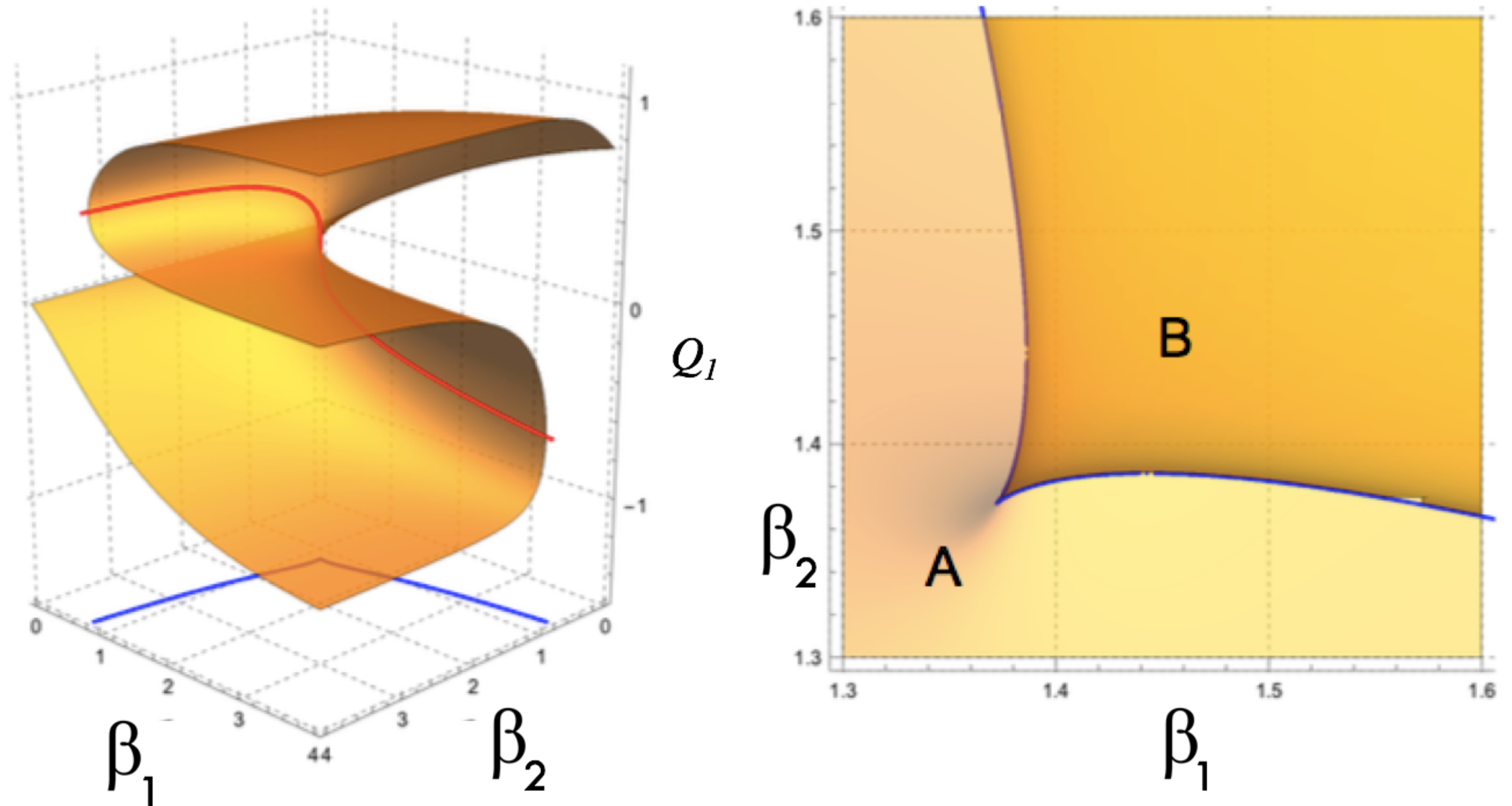}
\caption{\label{chicken_bif_set} Left Plot: The QRE surface for one of the agents in the {\it Chicken} game. The red curve on the surface is the set of critical points for the game, the blue curve in the $[\beta_1, \beta_2]$ plane is the projection of these critical points onto the control parameter space. Right Plot: View projecting down onto the control pane showing the regions of the control plane that have either 1 fixed point (A) or three fixed points (B).}
\end{figure}

Bifurcations in the QRE were first analysed by McKelvey et al~\cite{mckelvey1996statistical} (the {\it chicken game} discussed next) and a pitchfork bifurcation in the {\it battle of the sexes} game is covered in~\cite[Figure 6.3, page 153]{goeree2016quantal}. The equilibrium points in these plots were computed using the symbolic mathematics package Mathematica, however we note that there is also some very interesting work on computation of (Logit) QRE using homotopy methods. These methods can also be used to compute Nash equilibria, as limit points of the QRE, see~\cite{turocy2005dynamic}, cf. Section~\ref{comp_methods}. For our purposes the battle of the sexes serves to illustrate the bifurcations that can occur in the QRE as a covector of parameters $[\beta_1, \beta_2]$ are varied. The payoff bi-matrix for discrete strategies is given by:

{\renewcommand{\arraystretch}{2}
\begin{center}
    \begin{tabular}{cr|c|c|}
      & \multicolumn{1}{c}{} & \multicolumn{2}{c}{agent $2$}\\
      & \multicolumn{1}{c}{} & \multicolumn{1}{c}{Swerve}  & \multicolumn{1}{c}{Straight} \\\cline{3-4}
      \multirow{2}*{agent $1$}  & Swerve & $(0,0)$ & $(-1,+1)$ \\\cline{3-4}
      & Straight & $(+1,-1)$ & $(-10,-10)$ \\\cline{3-4}
    \end{tabular}
\end{center}}
\vspace{7mm}
This game represents interactions between two people, such as the game often seen in movies where drivers of two cars are heading directly towards one another. The challenge is for each driver to choose either {\it straight} or {\it swerve}, the driver who chooses straight wins and the driver who swerves loses. If they both choose to swerve the game is a draw, if they both choose straight both drivers crash into each other and lose. There are two pure strategy Nash equilibria: one where agent 1 swerves while agent 2 goes straight, the other where agent 1 goes straight while agent 2 swerves. These can be identified immediately by noting that, for either $[$swerve, straight$]$ or $[$straight, swerve$]$ neither agent can achieve a higher payoff by unilaterally changing their choice while the other agent's choice remains fixed. There is also a mixed strategy equilibrium. \\

When these game parameters are put into the QRE it can be shown that there is one fixed point for $[\beta_1, \beta_2] = [0,0]$ and three fixed points for $[\beta_1, \beta_2] = [\infty,\infty]$, these three fixed points correspond to the Nash equilibria of the Chicken game. In varying these parameters away from zero and towards $\infty$ there is necessarily some values of these parameters at which new fixed points emerge, a subset of these are shown in Figure~\ref{chicken_bif_set}. This bifurcation diagram was hinted at in McKelvey and Palfrey's original work and recent work by Wolpert and Harr\'e has explored some of the variety of this system of equations ~\cite{wolpert2011strategic,wolpert2012hysteresis,harre2013simple,harre2014strategic}.\\

To show this for $G_2^2$ games, first we rescale the probabilities so that: $Q_i = 2p_i^{j_i} - 1 \in [-1,1]$ for the probability of one of agent $i$'s choices so that we can work with a single variable $Q_i$ for each agent and we write the QRE as the functional relationships: 
\begin{eqnarray}
Q_i & = & f_i(Q_{-i}, \beta_i), \label{Qi1} \\
& = & f_i(f_{-i}(Q_{i},\beta_{-i}), \beta_i), \\
Q_{-i} & = & f_{-i}(Q_{i}, \beta_{-i}), \label{Qi2}\\
& = & f_{-i}(f_{i}(Q_{-i},\beta_i), \beta_{-i}).
\end{eqnarray} 
To find the set of singular points of the QRE surface we compute the Jacobian of the system by first finding all four terms of the form $\frac{\partial Q_i}{\partial \beta_k}$, for example~\cite{wolpert2011strategic}:
\begin{eqnarray}
\frac{\partial f_i}{\partial Q_{-i}}\frac{\partial f_{-i}}{\partial Q_i}\frac{\partial Q_i}{\partial \beta_i} + \frac{\partial f_i}{\partial \beta_i} - \frac{\partial Q_i}{\partial \beta_i} & = & 0 \\
f^1_i f^1_{-i} \frac{\partial Q_i}{\partial \beta_i} + f^2_i - \frac{\partial Q_i}{\partial \beta_i} & = & 0  \label{dummy} \\
\frac{f^2_i}{f^1_i f^1_{-i} -1} & = &  \frac{\partial Q_i}{\partial \beta_i} 
\end{eqnarray}
The simplification in notation at equation~\ref{dummy} uses the subscript to denote either the first or the second argument of $f_i$ in equations~\ref{Qi1} and~\ref{Qi2} that differentiation is with respect to.  A similar set of computations results in the Jacobian: 
\begin{eqnarray}
\mathbf{J}_{\beta}Q & = & \frac{1}{f^1_i f^1_{-i} -1}
\begin{bmatrix}
    f^2_{i}       & f^1_{i}f^2_{-i} \\
    f^2_{i}f^1_{-i}       & f^2_{-i}
\end{bmatrix}
\end{eqnarray}
which is singular when $f^1_i f^1_{-i} -1=0$. This set of solutions is shown as the red curve in the left plot of Figure~\ref{chicken_bif_set} and the projection of this set onto the control plane is shown in the right plot. When game theory is used as the basis for the construction of macroeconomic models of an economy this set is called the critical set of the economy. Sard's theorem allows us to conclude that the critical set of economies has measure zero~\cite{debreu1976regular,debreu1984economic} \\

\begin{figure}[h!t]
\center
\includegraphics[width=0.95\columnwidth]{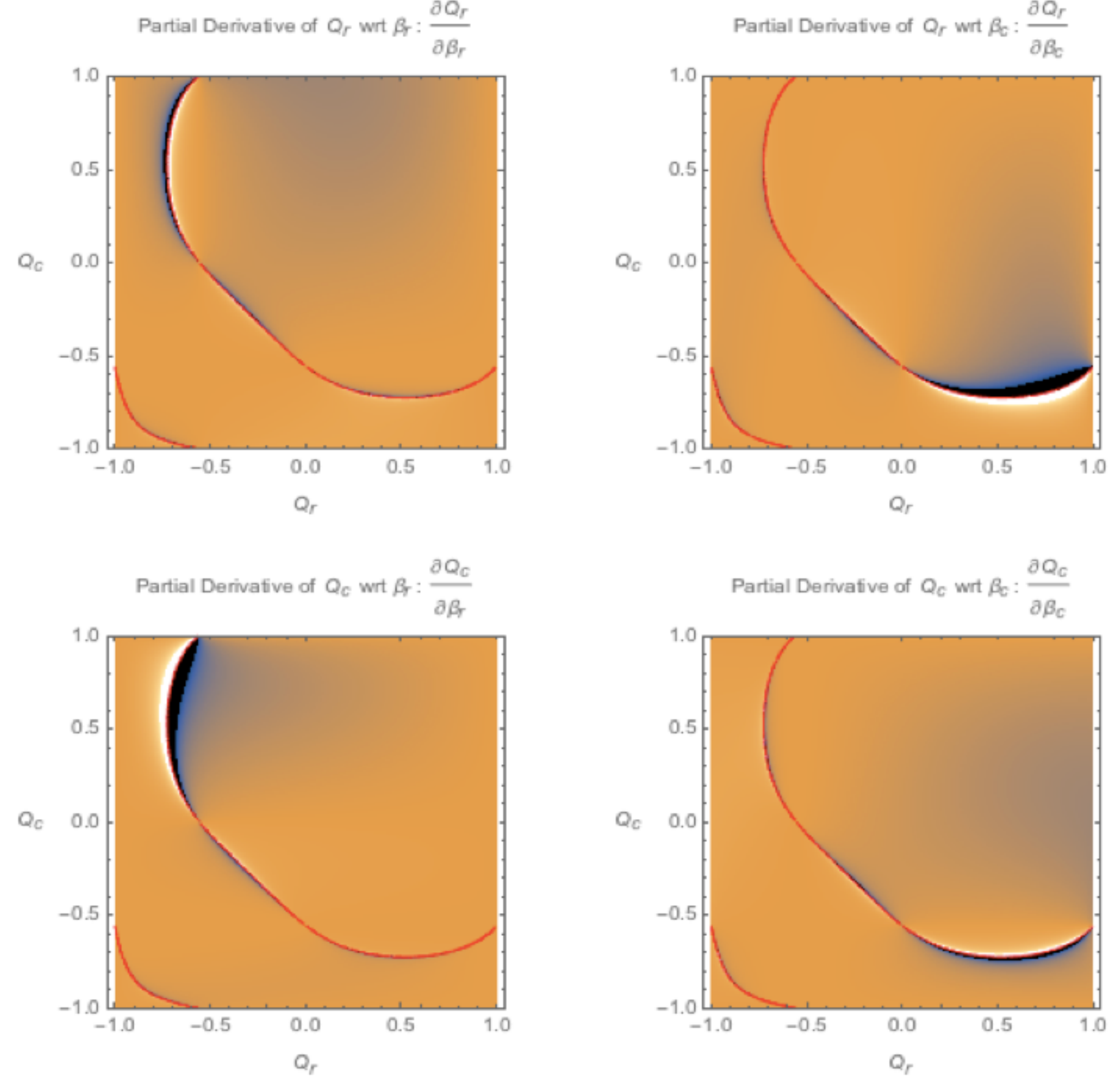}
\caption{\label{QQ_bif_set} The complete space of matrices $\mathbf{J}_{\beta}Q$ in which the colour represents the value of the gradient $\frac{\partial Q_i}{\partial \beta_j}$. The red curves are the set of {\it critical economies} and correspond to the red curves of Figure~\ref{chicken_bif_set}}. 
\end{figure}

Figure~\ref{chicken_bif_set} is one of the two mappings $f_j: \beta_i\times \beta_{-i}\rightarrow Q_j,$ $j\in \{i, -i\}$. An alternative representation is to consider the map $\mathbf{J}_{\beta}Q: Q_i \times Q_{-i} \rightarrow \mathbb{R}^{2\times 2}$ where  $\mathbb{R}^{2\times 2}$ is the family of matrices of partial derivatives $\frac{\partial Q_i}{\partial \beta_j} \in [-\infty, \infty]$ for $G_2^2$ games. This can be plotted directly as shown in Figure~\ref{QQ_bif_set} where we have used the Chicken game as an example. In order to generate these plots, a large number of explicit values for $\frac{\partial Q_i}{\partial \beta_j}$ were calculated and plotted, the results generate the background (golden-tan) colouring of the plots. Then the singularity set was calculated by implicitly solving the equation $f^1_i f^1_{-i} -1=0$ and overlaying these results on the plot of $\frac{\partial Q_i}{\partial \beta_j}$. It can be seen that where $\frac{\partial Q_i}{\partial \beta_j}$ diverges is where $f^1_i f^1_{-i} -1=0$. \\

One of the interesting points that arises in the representation of Figure~\ref{QQ_bif_set} is that, except possibly at finite and isolated points, the critical economic states partitions the $Q_i \times Q_{-i}$ space into three parts. An interpretation of this is that, if all you can observe of an economy is its distribution over states (the $Q_i$'s) then moving smoothly about this space is restricted by the partitions, there are states of the economy that are not accessible from a regular economy without passing through a critical state of the economy. \\ 


\section{Discussion}

The application of deterministic catastrophe theory to economics has a fraught history that has been well documented~\cite{rosser2007rise} but there are a number of articles that have been published recently that have begun a counter-swing in applications~\cite{rosser2013catastrophe,wilson2012catastrophe,jakimowicz2010catastrophes,diks2016can,kudinov2011catastrophes,barunik2015realizing} as well as some theoretical work~\cite{kudinov2011catastrophes}. In this article we have sought to briefly cover some of the recent material that has not elsewhere been collected in the hope that it might stimulate theoretical work to complement recent applied research.\\

We note that the original catastrophe theory, either deterministic or stochastic, is somewhat limited for direct economic application. The singularities that are traditionally studied require {\it a priori} assumptions on the potential function that may not be easily justified in an economic context. In particular, there is no agent-agent interaction explicit in the formulations of the potential function, and these interactions are central to economics. \\

On the other hand, the QRE has notable similarities with stochastic catastrophe theory. The similarities in the structures of equations~\ref{SDE_prob3} and~\ref{QRE_1} are notable, but their interpretation is not equivalent. Equation~\ref{SDE_prob3} is a single equation with a non-linear potential function whereas Equation~\ref{QRE_1} is a set of equations with linear interaction terms. \\

It might be tempting to assume that the exponentiated terms in these equations are both potentials, and a potential function is a necessary condition for the use of catastrophe theory, however this is also not immediately obvious. Ay {\it et al}~\cite{ay2017information} derived a potential function for game theory of the general form: 
\begin{eqnarray}
V(p) & = & \frac{1}{2}\sum_{i,j}a_{i,j}p_ip_j \label{game_potential}
\end{eqnarray}
where $A = a_{i,j}$ is a (constant valued) payoff matrix that describes all of the interactions between agents and the $p_i$ are the strategies of the agents. In this case, in order for $V(p)$ to be a potential function the following relationship needs to hold:
\begin{eqnarray}
a_{i,j} + a_{j,k} + a_{k,i} & = & a_{i,k} + a_{k,j} + a_{j,i} 
\end{eqnarray}
This condition is met if the matrix $A$ is symmetric. The relationship between these equations and their interpretations has yet to be explored. \\

These considerations have restricted the family of games to which catastrophe theory can be applied to those with symmetric payoff matrices. In economic theory, this is known to be satisfied by only a few games, in particular the {\it potential games} of Sandholm~\cite{sandholm2001potential} and a few other examples given in Rosser Jr.~\cite{rosser2007rise}. In general, the adoption of catastrophe theory in economics has been infrequent and very few examples beyond the cusp catastrophe have been studied (see~\cite{casti1975catastrophe} for a notable exception using the Butterfly catastrophe applied to housing markets). However, as Rosser Jr. has pointed out, there is considerable value to be had in adding catastrophe theory to the tool box of methods that economists use. More generally, there is a great deal of untouched territory in the formal analysis of the dynamics of economic systems near market crashes (singularities) in which there have been rapid and at times uncontrolled transitions between stable economic states in recent times. \\

{\bf Acknowledgement:} We are very thankful for the support of Laurentiu Paunescu and the diversity of people who were brought together for the JARCS workshop. M.S. Harr\'e was supported by ARC grant DP170102927.

\section*{Addendum: Some elementary remarks on the analytical foundations of Quantal equilibrium theory for $G_{2}^{2}$-games.}

\vspace{.2in}
This section will reiterate in the simplest possible terms the thread of mathematical ideas running through sections 3.3 and 3.4,
 connecting Game Theory as represented by elementary $G^{2}_{2}$-games to singularities, via the {\em quantal} approach to equilibria
proposed by McKelvey and Palfrey. Hence we will
consider only two players, each with the freedom to choose two strategies. The terms {\em player} and {\em strategy} (not
to mention {\em game}) have always been interpreted somewhat abstractly. The first of these could easily be replaced by
{\em agent} or {\em particle ensemble} without any loss of information. Similarly the word ``strategy" could easily be replaced 
by {\em state}. We will see that an appropriate choice of terms has everything to do with the manner in which equilibrium
is defined in relation to a specific mathematical model. Every $G_{2}^{2}$-game corresponds to a unique pair of $2\times 2$
utility matrices $U_{i}$ over the real numbers. The ``pure" strategies for each player are unit vectors $(1,0)$ and $(0,1)$
whereas ``mixed" strategies are pairs $\underline{x}^{i} = (x_{1}^{i}, x_{2}^{i})$ such that $0\leq x_{j}^{i}\leq 1$ and $x_{1}^{i}
+ x_{2}^{i} = 1 \ , \ i = 1,2$. The {\em utility functions} $g^{i}$ are then defined as a standard inner product
\[  g^{i}(\underline{x}^{1}, \underline{x}^{2}) = \langle \underline{x}^{1}, U_{i}\underline{x}^{2}\rangle \ .\]
Since $x^{i}_{2} = 1 - x^{i}_{1}$, these formulae reduce to a pair of quadratic functions $q_{i}(x,y)$ whose domain is 
restricted to the unit square $S$, such that $x = x^{1}_{1}$ and $y = x^{2}_{1}$. Now define 
\[\delta q_{1}(y) = q_{1}(1,y) - q_{1}(0,y)\hspace{.2in}\text{and}\hspace{.2in} \delta q_{2}(x) = q_{2}(x,1) - q_{2}(x,0) \ ,\]
and introduce non-decreasing, piecewise-continuous functions $F_{i}:{\Bbb R}\rightarrow [0,1]$, such that
\[\lim_{X\rightarrow -\infty}F_{i}(X) = 0\hspace{.2in}\text{and}\hspace{.2in}\lim_{X\rightarrow\infty}F_{i}(X) = 1 \hspace{.2in} (i = 1,2) \ .\]

Finally, let $\psi_{1}(x) = F_{1}(\delta q_{2}(x))$, $\psi_{2}(y) = F_{2}(\delta q_{1}(y))$, and define a map $\Psi : S\rightarrow S$ such that
$\Psi(x,y) = (\psi_{2}(y), \psi_{1}(x))$. An {\em equilibrium strategy} (or {\em equilibrium state}) will then correspond to a fixed point
\[ \Psi(x^*, y^*) = (x^*, y^*) \ .\]

Two special cases are of particular interest here. First, suppose the $F_{i}$ are defined in terms of the Heaviside function with parameter
$0\leq a\leq 1$: 
\[ H(X, a) = \left\{\begin{array}{ccc} 1 & if & X > 0\\
                                      0 & if & X < 0\\
                                      a & if & X = 0\\
\end{array}\right.\]
so that
\[\Psi(x,y) = (H(\delta q_{1}(y), x), H(\delta_{2}(x), y)) \ .\]
This in fact defines the ``best response" map in Nash's theory of equilibria. Note that
\[ \psi_{1}(x) = \left\{\begin{array}{ccc} 1 & if & q_{2}(x,1) > q_{2}(x,0)\\
                                           0 & if & q_{2}(x,1) < q_{2}(x,0)\\
\end{array}\right. \ ,\]
and 
\[ \psi_{2}(y) = \left\{\begin{array}{ccc} 1 & if & q_{1}(1,y) > q_{1}(0,y)\\
                                           0 & if & q_{1}(1,y) < q_{1}(0,y)\\
\end{array}\right. \ .\]
Since the $\delta q_{i}$ are linear functions, these relations indicate that apart from the pure-strategy Nash-equilibria 
to be found among the vertices of $S$, a unique mixed-strategy Nash equilibrium occurs 
precisely when $\delta q_{1}(y) = \delta q_{2}(x) = 0$. The choice of $F$ is well-suited to the classical conception of
game theory, in which equilibria are ``rationally optimized" strategies available to both players. 

We turn now to a second model, in which players do not rationally fix their strategy in response to that of their opponent.
Instead, strategies are governed by a smooth probability distribution. Applied to a large sample space of random trials of 
a specific game, this model may
or may not realistically reflect the behaviour of human populations, though it is naturally adapted from numerically large particle-ensembles of 
the kind encountered in statistical thermodynamics. A standard probability distribution associated with the partition function which lies at 
the foundation of this theory is given as
\[ F_{i}(X) = \frac{1}{1 + e^{-\beta_{i}X}} \hspace{.2in} (i = 1,2) \ ,\]
where the parameters $\beta_{i}$ are strictly positive real numbers. Note that
\[ \psi_{1}(x) = \frac{1}{1 + e^{-\beta_{1}\delta q_{2}(x)}} = \frac{e^{\beta_{1}q_{2}(x,1)}}{e^{\beta_{1}q_{2}(x,1)} + e^{\beta_{1}q_{2}(x,0)}} \ ,\]
and similarly 
\[ \psi_{2}(y) = \frac{e^{\beta_{2}q_{1}(1,y)}}{e^{\beta_{2}q_{1}(1,y)} + e^{\beta_{2}q_{1}(0,y)}} \ .\]

$\psi_{1}(x)$ may be interpreted as the probability $y$ that ensemble 2 lies in state 1, given that the probability of ensemble 1 lying in the same state 
is $x$. Conversely $\psi_{2}(y)$ represents the probability $x$ that ensemble 1 lies in state 1, given a probability $y$ that ensemble 2 lies in the same state.
Given two events $A$ and $B$, the symmetry between the classical definitions of conditional probabilities $P(A | B)$ and $P(B | A)$ simply occurs when 
$P(A) = P(B)$, where in general the probability of one event is construed as conditional on the actual occurrence of the other, as implied by Bayes' Theorem.
By contrast, the model above estimates the probability of one event as conditional on a given {\em probability} of the other, as is traditionally postulated in the world 
of quantum interactions. Symmetry occurring in the relations $y = \psi_{1}(x)$
and $x = \psi_{2}(y)$ is then another way of characterizing points of equilibrium. Since $S$ is a closed and bounded subset of ${\Bbb R}^{2}$, and in the present situation
$\Psi: S\rightarrow S$ is a continuous map, the existence of fixed points is implied by the fact that any sequence of iterations $\{\Psi(x_{k}, y_{k})\}_{k=1}^{\infty}$
where $(x_{k}, y_{k}) = \Psi(x_{k-1}, y_{k-1})$, must have a convergent subsequence. This is essentially the Fixed Point Theorem of Brouwer. Note, however, that if
the distributions $F_{i}$ are not strictly continuous, as in the case of Nash equilibria, then the more general theorem of Kakutani may still be applied. With respect to 
an appropriately chosen norm on function-space, the Heaviside distribution may in fact be recovered as an asymptotic $\beta$-limit of the smooth distribution above.
 For generic values of $(\beta_{1}, \beta_{2})$ the number 
and location within $S$ of the fixed points of $\Psi$ (for a given choice of the matrices $U_{i}$) will depend smoothly on these parameters, which suggests that the set
of fixed points for each $(\beta_{1}, \beta_{2})$ will generate a possibly branched topological covering of the parameter-plane, corresponding to the surface
\[ \Sigma = \{x = \psi_{2}(y)\}\cap\{y = \psi_{1}(x)\} \subset{\Bbb R}_{+}^{2}\times S \ ,\]  
where ${\Bbb R}_{+} = (0,\infty)$. Now let 
\[ f_{1}(x,\beta_{1}, \beta_{2}) = x - \psi_{2}(\psi_{1}(x))\hspace{.2in}\text{and}\hspace{.2in} f_{2}(y,\beta_{1}, \beta_{2}) = y - \psi_{1}(\psi_{2}(y)) \ .\]
According to the Implicit Function Theorem,
\[ \frac{\partial f_{1}}{\partial x}(x^{*}, \beta_{1}^{*}, \beta_{2}^{*}) = 1 - \frac{\partial\psi_{2}}{\partial y}(\psi_{1}(x^{*}))\frac{\partial\psi_{1}}{\partial x}(x^{*})\neq 0\]
implies the existence of a function $x = \varphi_{1}(\underline{\beta})$ in a neighbourhood of $\underline{\beta}^{*}$. Similarly, 
\[ \frac{\partial f_{2}}{\partial y}(y^{*}, \beta_{1}^{*}, \beta_{2}^{*}) = 1 - \frac{\partial\psi_{1}}{\partial x}(\psi_{2}(y^{*}))\frac{\partial\psi_{2}}{\partial y}(y^{*})\neq 0\]
implies the existence of a function $y = \varphi_{2}(\underline{\beta})$ in a neighbourhood of $\underline{\beta}^{*}$. Hence, in a neighbourhood of any regular 
equilibrium point $(\underline{\beta}
^{*}, x^{*}, y^{*})$, the surface $\Sigma$ is parametrized smoothly by functions $\varphi = (\varphi_{1}, \varphi_{2})$. 

Conversely, for every $(\beta_{1}, \beta_{2}, x^{*}, y^{*})\in\Sigma$ belonging to the critical locus
\[ \Gamma = \{1 - \frac{\partial\psi_{1}}{\partial x}(x^{*})\frac{\partial\psi_{2}}{\partial y}(y^{*}) = 0\} \ ,\]
the standard projection $\pi:{\Bbb R}^{2}\times S\rightarrow{\Bbb R}^{2}$ maps $\Gamma\cap\Sigma$ to the ``branch locus"
\[\pi(\Gamma\cap\Sigma)\subset{\Bbb R}_{+}^{2} \ .\]

In this formulation of equilibrium without reference to a potential function, there is no apparent classification of critical loci
in terms of Thom's elementary catastrophes. The projection $\pi$ above suggests rather that the branch locus might be understood via normal forms
of {\em generic mappings} between smooth surfaces, as in the classic Theorem of Whitney (cf., e.g.~\cite{golubitsky2012stable}). This would seem to be the case for games of the sort represented in Fig. 4, but the general situation is not immediately clear. In particular, it must be asked whether all
mappings $\pi\mid_{\Sigma}$ which arise in the context above for $G_{2}^{2}$-games are in fact generic, in the sense that their first jet extension is
always transversal to the corank-one submanifold inside the jet space $J^{1}(\Sigma, {\Bbb R}_{+}^{2})$.

\vspace{.2in}

\bibliographystyle{plain}
\bibliography{Singularities_bib}

\begin{thebibliography}{10}

\bibitem{afriat2014demand}
Sydney~N Afriat.
\newblock {\em Demand Functions and the Slutsky Matrix.(PSME-7)}, volume~7.
\newblock Princeton University Press, 2014.

\bibitem{arnol2003catastrophe}
Vladimir~I Arnol'd.
\newblock {\em Catastrophe theory}.
\newblock Springer Science \& Business Media, 2003.

\bibitem{arnold1984singularities}
Vladimir~Igorevich Arnold.
\newblock Singularities in optimization problems, the maxima function.
\newblock In {\em Catastrophe Theory}, pages 43--46. Springer, 1984.

\bibitem{ay2017information}
Nihat Ay, J{\"u}rgen Jost, H{\^o}ng V{\^a}n~L{\^e}, and Lorenz
  Schwachh{\"o}fer.
\newblock {\em Information geometry}, volume~64.
\newblock Springer, 2017.

\bibitem{barunik2015realizing}
Jozef Barunik and Jiri Kukacka.
\newblock Realizing stock market crashes: stochastic cusp catastrophe model of
  returns under time-varying volatility.
\newblock {\em Quantitative Finance}, 15(6):959--973, 2015.

\bibitem{barunik2009can}
Jozef Barun{\'\i}k and M~Vosvrda.
\newblock Can a stochastic cusp catastrophe model explain stock market crashes?
\newblock {\em Journal of Economic Dynamics and Control}, 33(10):1824--1836,
  2009.

\bibitem{berry1982universal}
MV~Berry.
\newblock Universal power-law tails for singularity-dominated strong
  fluctuations.
\newblock {\em Journal of Physics A: Mathematical and General}, 15(9):2735,
  1982.

\bibitem{casti1975catastrophe}
John Casti and Harry Swain.
\newblock Catastrophe theory and urban processes.
\newblock In {\em IFIP Technical Conference on Optimization Techniques}, pages
  388--406. Springer, 1975.

\bibitem{cobb1978stochastic}
Loren Cobb.
\newblock Stochastic catastrophe models and multimodal distributions.
\newblock {\em Systems Research and Behavioral Science}, 23(4):360--374, 1978.

\bibitem{codenotti2011computational}
Bruno Codenotti.
\newblock Computational game theory, 2011.

\bibitem{dasgupta2015debreu}
Partha~Sarathi Dasgupta and Eric~S Maskin.
\newblock Debreu's social equilibrium existence theorem.
\newblock {\em Proceedings of the National Academy of Sciences},
  112(52):15769--15770, 2015.

\bibitem{debreu1970economies}
Gerard Debreu.
\newblock Economies with a finite set of equilibria.
\newblock {\em Econometrica: Journal of the Econometric Society}, pages
  387--392, 1970.

\bibitem{debreu1976regular}
Gerard Debreu.
\newblock Regular differentiable economies.
\newblock {\em The American Economic Review}, 66(2):280--287, 1976.

\bibitem{debreu1984economic}
Gerard Debreu.
\newblock Economic theory in the mathematical mode.
\newblock {\em The Scandinavian Journal of Economics}, 86(4):393--410, 1984.

\bibitem{debreu1987theory}
Gerard Debreu.
\newblock {\em Theory of value: An axiomatic analysis of economic equilibrium}.
\newblock Number~17. Yale University Press, 1987.

\bibitem{diks2016can}
Cees Diks and Juanxi Wang.
\newblock Can a stochastic cusp catastrophe model explain housing market
  crashes?
\newblock {\em Journal of Economic Dynamics and Control}, 69:68--88, 2016.

\bibitem{eaves1971linear}
B~Curtis Eaves.
\newblock The linear complementarity problem.
\newblock {\em Management science}, 17(9):612--634, 1971.

\bibitem{goeree2016quantal}
Jacob~K Goeree, Charles~A Holt, and Thomas~R Palfrey.
\newblock {\em Quantal response equilibria}.
\newblock Springer, 2016.

\bibitem{golubitsky2012stable}
Martin Golubitsky and Victor Guillemin.
\newblock {\em Stable mappings and their singularities}, volume~14.
\newblock Springer Science \& Business Media, 2012.

\bibitem{guckenheimer1973catastrophes}
John Guckenheimer.
\newblock Catastrophes and partial differential equations.
\newblock {\em Ann. Inst. Fourier}, 23(2):31, 1973.

\bibitem{harre2013simple}
Michael~S Harr{\'e}, Simon~R Atkinson, and Liaquat Hossain.
\newblock Simple nonlinear systems and navigating catastrophes.
\newblock {\em The European Physical Journal B}, 86(6):1--8, 2013.

\bibitem{harre2014strategic}
Michael~S Harr{\'e} and Terry Bossomaier.
\newblock Strategic islands in economic games: Isolating economies from better
  outcomes.
\newblock {\em Entropy}, 16(9):5102--5121, 2014.

\bibitem{hauert2001fundamental}
Ch~Hauert.
\newblock Fundamental clusters in spatial 2$\times$ 2 games.
\newblock {\em Proceedings of the Royal Society of London B: Biological
  Sciences}, 268(1468):761--769, 2001.

\bibitem{herings2010homotopy}
P~Jean-Jacques Herings and Ronald Peeters.
\newblock Homotopy methods to compute equilibria in game theory.
\newblock {\em Economic Theory}, 42(1):119--156, 2010.

\bibitem{jakimowicz2010catastrophes}
A~Jakimowicz.
\newblock Catastrophes and chaos in business cycle theory.
\newblock {\em Acta Physica Polonica, A.}, 117(4), 2010.

\bibitem{kudinov2011catastrophes}
AN~Kudinov, VP~Tsvetkov, and IV~Tsvetkov.
\newblock Catastrophes in the multi-fractal dynamics of social-economic
  systems.
\newblock {\em Russian Journal of Mathematical Physics}, 18(2):149--155, 2011.

\bibitem{lemke1964equilibrium}
Carlton~E Lemke and Joseph~T Howson, Jr.
\newblock Equilibrium points of bimatrix games.
\newblock {\em Journal of the Society for Industrial and Applied Mathematics},
  12(2):413--423, 1964.

\bibitem{mckelvey1996computation}
Richard~D McKelvey and Andrew McLennan.
\newblock Computation of equilibria in finite games.
\newblock {\em Handbook of computational economics}, 1:87--142, 1996.

\bibitem{mckelvey1995quantal}
Richard~D McKelvey and Thomas~R Palfrey.
\newblock Quantal response equilibria for normal form games.
\newblock 1995.

\bibitem{mckelvey1996statistical}
Richard~D McKelvey and Thomas~R Palfrey.
\newblock A statistical theory of equilibrium in games.
\newblock {\em The Japanese Economic Review}, 47(2):186--209, 1996.

\bibitem{nash1951non}
John Nash.
\newblock Non-cooperative games.
\newblock {\em Annals of mathematics}, pages 286--295, 1951.

\bibitem{nash1950equilibrium}
John~F Nash.
\newblock Equilibrium points in n-person games.
\newblock {\em Proceedings of the national academy of sciences}, 36(1):48--49,
  1950.

\bibitem{rosser2007rise}
J~Barkley Rosser.
\newblock The rise and fall of catastrophe theory applications in economics:
  Was the baby thrown out with the bathwater?
\newblock {\em Journal of Economic Dynamics and Control}, 31(10):3255--3280,
  2007.

\bibitem{rosser2013catastrophe}
J~Barkley Rosser.
\newblock {\em From catastrophe to chaos: a general theory of economic
  discontinuities}.
\newblock Springer Science \& Business Media, 2013.

\bibitem{sandholm2001potential}
William~H Sandholm.
\newblock Potential games with continuous player sets.
\newblock {\em Journal of Economic theory}, 97(1):81--108, 2001.

\bibitem{scarf1967approximation}
Herbert Scarf.
\newblock The approximation of fixed points of a continuous mapping.
\newblock {\em SIAM Journal on Applied Mathematics}, 15(5):1328--1343, 1967.

\bibitem{shapley1974note}
Lloyd~S Shapley.
\newblock A note on the lemke-howson algorithm.
\newblock In {\em Pivoting and Extension}, pages 175--189. Springer, 1974.

\bibitem{stewart1977catastrophe}
Ian Stewart.
\newblock Catastrophe theory.
\newblock {\em Math. Chronicle}, 5:140--165, 1977.

\bibitem{thom1975structural}
Rene Thom.
\newblock Structural stability and morphogenesis, 1975.
\newblock {\em Trans. by D. Fowler. Reading, Mass.: Benjamin}, 1975.

\bibitem{thom1976structural}
Ren{\'e} Thom.
\newblock Structural stability and morphogenesis, 1976.

\bibitem{thom1977structural}
Ren\'e Thom.
\newblock Structural stability, catastrophe theory, and applied mathematics.
\newblock {\em SIAM review}, 19(2):189--201, 1977.

\bibitem{turocy2005dynamic}
Theodore~L Turocy.
\newblock A dynamic homotopy interpretation of the logistic quantal response
  equilibrium correspondence.
\newblock {\em Games and Economic Behavior}, 51(2):243--263, 2005.

\bibitem{wagenmakers2005transformation}
Eric-Jan Wagenmakers, Peter~CM Molenaar, Raoul~PPP Grasman, Pascal~AI
  Hartelman, and Han~LJ van~der Maas.
\newblock Transformation invariant stochastic catastrophe theory.
\newblock {\em Physica D: Nonlinear Phenomena}, 211(3):263--276, 2005.

\bibitem{wilson2012catastrophe}
Alan Wilson.
\newblock {\em Catastrophe Theory and Bifurcation (Routledge Revivals):
  Applications to Urban and Regional Systems}.
\newblock Routledge, 2012.

\bibitem{wolpert2011strategic}
David Wolpert, Julian Jamison, David Newth, and Michael Harr\'e.
\newblock Strategic choice of preferences: the persona model.
\newblock {\em The BE Journal of Theoretical Economics}, 11(1), 2011.

\bibitem{wolpert2012hysteresis}
David~H Wolpert, Michael Harr{\'e}, Eckehard Olbrich, Nils Bertschinger, and
  Juergen Jost.
\newblock Hysteresis effects of changing the parameters of noncooperative
  games.
\newblock {\em Physical Review E}, 85(3):036102, 2012.

\bibitem{zeeman1974unstable}
E~Christopher Zeeman.
\newblock On the unstable behaviour of stock exchanges.
\newblock {\em Journal of mathematical economics}, 1(1):39--49, 1974.

\end{thebibliography}

\end{document}